\documentclass[pra, reprint, aps,superscriptaddress]{revtex4-1}

\usepackage{amsmath}
\usepackage{amssymb}
\usepackage{bbold}
\usepackage{amsthm}
\usepackage{graphicx}
\usepackage{cleveref}
\usepackage{dsfont}
\usepackage{xcolor}
\usepackage{tikz}

\DeclareMathOperator{\tr}{Tr}

\newcommand{\id}{\mathbb{1}}

\newcommand{\dens}[1]{|#1\rangle\langle#1|}

\newcommand{\ket}[1]{|#1\rangle}
\newcommand{\bra}[1]{\langle#1|}

\newcommand{\mc}[1]{\mathcal{#1}}

\newcommand{\md}[1]{\mathds{#1}}

\newcommand{\ct}{^\dagger}

\newtheorem{lemma}{Lemma}

\newcounter{notecounter}

\newcommand{\LAB}{\ensuremath{\Lambda_{A\!\rightarrow \!B}}}
\newcommand{\LBA}{\ensuremath{\Lambda_{B\!\rightarrow \!A}}}

\usepackage{algorithm}
\usepackage{algpseudocode}
 \usepackage{relsize}
\usepackage{silence}
\begin{document}
\author{Jonas Helsen}
\affiliation{Qusoft \& Korteweg-de Vries Instituut, University of Amsterdam, Science Park 123 1098 XG Amsterdam, The Netherlands}
\author{Stephanie Wehner}
\affiliation{QuTech, Delft University of Technology, Lorentzweg 1, 2628 CJ Delft, The Netherlands}
\affiliation{Kavli Institute of Nanoscience, Delft University of Technology, Lorentzweg 1, 2628 CJ Delft, The Netherlands}
\title{A benchmarking procedure for quantum networks}

\date{\today}
\begin{abstract}
\noindent We propose network benchmarking: a procedure to efficiently benchmark the quality of a quantum network link connecting quantum processors in a quantum network. This procedure is based on the standard randomized benchmarking protocol and provides an estimate for the fidelity of a quantum network link. We provide statistical analysis of the protocol as well as a simulated implementation inspired by NV-center systems using Netsquid, a special purpose simulator for noisy quantum networks.
\end{abstract}
\maketitle
\section{Introduction}\label{sec:introduction}
Quantum technology research can be broadly categorized into two strands: on the one hand the development of large-scale fault-tolerant quantum computers, and on the other hand the development of quantum networks that link quantum computers together and allow for quantum communication based tasks (such as clock synchronization \cite{jozsa2000quantum}, anonymous communication~\cite{christandl2005quantum}, and cryptography~\cite{bennett1992experimental,ekert1991quantum}), culminating in a Quantum Internet~\cite{wehner2018quantum} connecting quantum processing nodes. These nodes, for which physical platforms such as NV-centers in diamond~\cite{reiserer2016robust}, ion traps~\cite{inlek2017multispecies}, and neutral atoms~\cite{welte2018photon} are currently being developed, posses quantum computing capacity, leading to the possibility of distributed or networked quantum computing~\cite{spiller2006quantum}.

One of the major step changes in the development of quantum computers in the last decade has been the development of practical methods to characterize the quality of quantum operations, allowing experimentalists to quickly diagnose and improve a critical building block of a fault-tolerant quantum computing architecture (see e.g. \cite{eisert2020quantum} and references therein). 
In this work we consider the corresponding problem of the characterization of quantum communication links, a key feature of quantum networks that has no real counterpart in quantum computation. Several methods exist to assess the quality of a quantum network link which we briefly review. For entanglement based networks, i.e. networks where the quantum network link is established through entangled states between nodes, any characterization of the quality of entanglement translates in principle to a quality measure of the network link. Many methods exist to assess the quality of entanglement (see e.g. \cite{brunner2014bell,vsupic2020self} for work on Bell inequalities and self-testing, and \cite{d20048} for quantum state tomography), which can be mapped to quality-assessment methods for quantum network links.  Similarly, for direct transmission based network links (abstractly modeled by a quantum channel), we may, from the transmission of qubit states in two distinct bases (typically the $X$ and $Z$ bases) make an inference about how well any state, or indeed entanglement, may be transmitted (see e.g.\cite{barnum2000quantum,buscemi2010quantum,chuang1997prescription}). More generally, a procedure is known to estimate the capacity of a quantum channel \cite{pfister2018capacity}). Finally, \cite{lipinska2020certification} gives a method certify whether a quantum network of nodes connected by quantum links has attained a specific stage of development.



In this work we aim to add to this toolbox by proposing network benchmarking: a procedure to assess the quality of transmission between quantum processing nodes in a quantum network in the so-called \emph{quantum memory network stage} and above~\cite{wehner2018quantum}, by yielding an estimate of the average fidelity of the effective quantum channel modeling a quantum network link. 
Network benchmarking is adapted from the randomized benchmarking protocol~\cite{PhysRevA.75.022314,knill2008randomized}, a gold-standard methodology for the characterization of quantum operations in quantum computers. It is lightweight, easy to implement, and inherits many of the robustness properties enjoyed by the original randomized benchmarking protocol. We also give a more general multi-node protocol that can be used to characterize the fidelity of a path of multiple nodes connected by quantum communication links, and can thus be seen as the quantum analogon to the classical 'ping' operation.

\subsection{Results}
We introduce network benchmarking, a method that robustly and efficiently yields an estimate of the quality of a quantum network link. We propose two version of this protocol: a $2$-node protocol that estimates the quality of a link between two quantum network nodes, and a multi-node protocol that estimates the quality of a path over several nodes in a network. We provide a theoretical analysis of these protocols, arguing that they estimate the average fidelity of the quantum channel modeling a quantum link. For network links implemented by noisy quantum teleportation we prove that this network fidelity can be related to the \emph{average fidelity}, a standard metric of quality of quantum processes.\\

We supplement this theoretical work with numerical simulations using the quantum network simulator Netsquid~\cite{netsquid}. By testing several realistic scenarios we can argue that network benchmarking performs well under realistic conditions (e.g. noise, timing, circuit decompositions), efficiently yielding accurate estimates of the network fidelity of a network link. The code for these simulations can be found at~\cite{zenodolink}.\\

Finally we analyze the statistical requirements of network benchmarking, with a particular focus on the number of times the quantum communication link must be used to get a good estimate of the average fidelity.\\ 

In \cref{sec:prelims} we elaborate upon a model of a quantum network and recall aspects of the quantum channel formalism for noisy quantum operations. In \cref{sec:net-bench} we  introduce the network benchmarking protocol, in its $2$-node and multi-node versions, and in \cref{sec:fidelity} we connect the data it generates to the average fidelity. In \cref{sec-sim} we present results from numerical simulations of the network benchmarking protocol using the NetSquid simulation package for quantum networks. In \cref{sec:stat-anal} we discuss the statistics of network benchmarking.

\begin{figure*}[t]
\includegraphics[width=2\columnwidth]{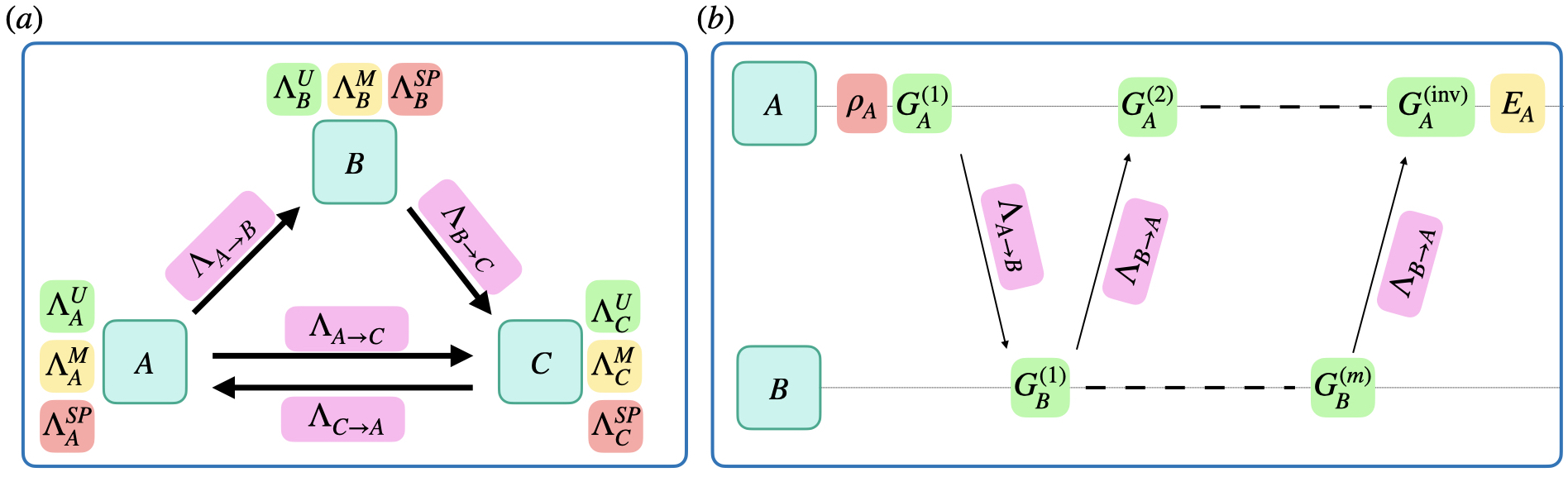}
\caption{{\bf (a)} An example network with three nodes $A,B,C$ presented together with all associated quantum channels modeling state preparation noise ($\Lambda^{SP}$), measurement noise $\Lambda^M$, operation noise ($\Lambda^U$), and modeling noise in the network link connecting the node ( e.g. $\LAB$ for the link $A
\to B$). {\bf (b)} Graphical description of (a single run of $m$ bounces of) the $2$-node network benchmarking protocol, with time running rightwards. Boxes indicate actions taken at nodes $A,B$ and colors (color online) associated to each box indicate what noise process (see (a)) affects these actions. See \cref{fig:protocol2nodes} for a detailed description of the $2$-node network benchmarking protocol. }\label{fig:node_prot}
\end{figure*}
\section{Preliminaries}\label{sec:prelims}
In this section we will introduce the necessary mathematical machinery to state our main results. We will discuss our model of a quantum network, as well as tools for modeling noise in quantum operations.
\subsection{Network model}
We consider an abstract model of a quantum network, consisting of nodes and connections between those nodes. We will label the nodes by capital letters ($A,B,..$) and denote the connection between two nodes by a directional arrow ( $A\to B$). These nodes and connections are abstractions of the physical components in the network. Within the framework of \cite{wehner2018quantum} we will assume that our nodes have the following functionalities: \\

(1) The ability to store quantum states in memory (stage $4$ in \cite{wehner2018quantum}). We model this by associating a memory register $\mc{H}_A$ to each node. We will assume this register can be initialized in some fixed initial state $\rho_A$ and read out by measurement in a POVM $\{E_A^{(i)}\}_{i\in I}$ with $I$ being some index set labeling the possible measurement outcomes. An example of an initial state is the all-zero $\ket{0\ldots 0}$ state and an example of a POVM measurement is the standard computational basis measurement.\\

(2) The ability to perform quantum operations on stored quantum states (stage $5$ in \cite{wehner2018quantum}). We will model this by allowing the application of quantum gates $U$ from a gateset $\md{G}$. Ideally this gateset is universal, meaning that any unitary operation can be implemented on the quantum processing node by sequences of unitaries from $\md{G}$, but we will only need a weaker property to perform network benchmarking, as we shall see in \cref{sec:net-bench}.\\

(3) The ability to transmit quantum states from a node $A$ to a node $B$ (stage $3$ in \cite{wehner2018quantum}). This can be implemented in various ways in the underlying hardware, such as through teleportation using entanglement, but we will consider it as an abstract functionality here. It is this ability that is predominantly tested by network benchmarking. 

\subsection{Noise and average fidelity}
The appearance of noise in quantum devices is typically modeled by quantum channels. These are superoperators (linear maps that send matrices to matrices), that preserve physicality, i.e. they map quantum states to quantum states. For an extended introduction to quantum channel see \cite{nielsen2002quantum}. We will denote quantum channels by $\Lambda$ and denote the action of a quantum channel on a state $\rho$ as $\Lambda(\rho)$. We will use superscripts to indicate the function of a quantum channel and subscripts to denote the node to which they are associated. For instance, we will model the noise associated to state preparation in node $A$ by $\Lambda^{SP}_{A}$ and the noise associated to measurement by $\Lambda^{M}_{A}$. We will also associate to each quantum operation $U$ (which we think of a as a superoperator $U(\rho) = U\rho U\ct$, abusing notation somewhat) implemented on a node $A$ a quantum channel $\Lambda_A^U$ modeling the noise associated to the operation $U$. This means that if node $A$ is instructed to apply $U$ to a state $\rho$ (yielding $(U(\rho) = U\rho U\ct)$) it actually outputs $(\Lambda^U_AU)(\rho) = \Lambda_A^U(U\rho U\ct)$. Note that we have made no mention of the physical mechanism by which $\Lambda_A^U$ arises, it is abstracted away. Finally we associate to the quantum transmission link $A\to B$ the quantum channel $\Lambda_{A\to B}$, modeling the noise incurred by a state in the process of transmission between $A$ and $B$. Note that in general we do not assume that $\LAB = \LBA$, although we will see examples where this is the case. We have included an illustration of the network model and associated noise maps in \cref{fig:node_prot}(a). \\

Upon modeling (the act of transmission through) a network link $A\to B$ with a quantum channel, we can ask how we can quantify the quality of this network link. This question is equivalent to asking how well the quantum channel $\LAB$ approximates the identity channel. One of the most common ways of quantifying this approximation is the average fidelity $F(\Lambda)$ of a quantum channel $\Lambda$, which is defined as
\begin{equation}\label{eq:av_fid}
F(\Lambda) = \int d\psi \tr \big[\Lambda\big(\dens{\psi}\big)\dens{\psi}\big],
\end{equation}
where the integral is taken uniformly over all pure quantum states $\ket{\psi}$. One can interpret this quantity as measuring how much a generic quantum state changes when $\Lambda$ acts on it, or equivalently  as capturing the average behavior of $\Lambda$. The average fidelity is a standard metric used in reporting the quality of quantum operations in quantum computers. The goal of network benchmarking is to estimate quantities like $F(\LAB)$, the average fidelity of (the quantum channel modeling) a network link $A\to B$. \\

Finally we note that the average fidelity is closely related to another quantity~\cite{nielsen2002simple}, which we call the depolarizing fidelity $f(\Lambda)$. These two quantities are related as
\begin{equation}\label{eq:fid_to_dep}
F(\Lambda) = \frac{(d-1)f(\Lambda) +1}{d},
\end{equation}
where $d$ is the dimension of the underlying state space. The depolarizing fidelity does not have the clean operational interpretation of the average fidelity, but it will show up more naturally in the calculations below. Next we move on to defining and analyzing the network benchmarking protocol.

\section{Network benchmarking}\label{sec:net-bench} 
In this section we introduce network benchmarking. We will describe two versions of this protocol, a $2$-node protocol and a more general multi-node protocol. Network benchmarking can be seen as an adaption of the randomized benchmarking protocol \cite{knill2008randomized,PhysRevA.75.022314,helsen2019multiqubit} for quantum networks, and will share many of its characteristics and theoretical analysis.\\

Consider two separated nodes $A$ and $B$ connected by a quantum network links $A\rightarrow B$ and $B\to A$, with associated quantum channels $\LAB,\LBA$. The goal of $2$-node network benchmarking is to estimate the average fidelities $F(\LAB)$ and $F(\LBA)$. However we desire that the procedure estimating these quantities satisfies several properties. The first property is \emph{efficiency}: we demand that the estimation procedure is light on resource use (measured in the number of times a network link is used), and independent of the capacity of the network link. By this we mean that we want to be able to estimate the fidelity of links sending many-qubit states in parallel without an exponential explosion in resource use. The second property is \emph{resistance to state preparation and measurement errors}. We will not assume that the initialization of states and the measuring of POVMs in nodes is perfect, and we demand that that estimation procedure output the correct result even when state preparation and measurement (SPAM) is imperfect. Ideally we would also like to demand independence of noise in quantum operations performed locally, but this is not possible. However, given that gate fidelities are typically much higher than state preparation and measurement fidelities in many physical platforms for networks nodes demanding only SPAM-robustness is a reasonable compromise. \\

Network benchmarking is not a device independent protocol, and in order to guarantee that it outputs an estimate of the fidelity of the quantum network link we have to make several assumptions on the behaviour of the nodes and network link. These assumptions are essentially the same as those of standard randomized benchmarking, see \cite{helsen2020general} for a general discussion of these assumptions. The central assumption we make is that of Markovianity: we will assume that the noise in the network link is always modeled by the same quantum channel, independent of the history of its use. We will similarly assume that the noise on state preparation $\rho_A$, measurements $\{E_A^{(i)}\}_{i\in I}$, and quantum gates $U_A$, have noise models that only depend on the node $A$ (and not on external variables like time, history, etc..). Note that this assumption of Markovianity was already implicit in our earlier description of the network model.\\

We will also assume that the quantum gates $U_A$ has a so-called gate-independent noise model. This means we assume that there exists a quantum channel $\Lambda_A$ such that for all gates $U\in \md{G}$ the implementation of $U$ is given by $\Lambda_A(U\rho U\ct))$. We stress however that this merely a technical assumption -standard in the randomized benchmarking literature-, adopted to make the proof of correctness of network benchmarking easier to understand. It can be removed at the cost of a considerable increase in mathematical complexity, see \cite{helsen2020general} for a general treatment.

\subsection{2-node network benchmarking}\label{subsec:protocol}
The $2$-node network benchmarking protocol involves two nodes $A$ and $B$ connected by links $A\to B$ and $B\to A$. This protocol produces an estimation of the (geometric) mean quality of the quantum channels $\LAB, \LBA$ associated to the links.
A formal specification of the 2-node network benchmarking protocol is given in \ref{fig:protocol2nodes}. An illustration of the steps of the protocol can also be found in \cref{fig:node_prot}(b). Here we give a more intuitive explanation the steps taken.

The protocol begins with the initialization of a state $\rho_A$ at node $A$. To this state a quantum operation $G_A^{(1)}$ is applied and the resulting state is then sent (through $\LAB$) to node $B$. Upon arrival at $B$ another quantum operation $G_B^{(1)}$ is applied and the state is sent back to node $A$ (through $\LBA$). The quantum operations $G_A^{(1)},G_B^{(1)}$ must be chosen at random from a sufficiently large set of quantum operations $\md{G}$. By sufficiently large we mean that the set must be at least a unitary $2$-design. A common choice for such a set is the multi-qubit Clifford group $\mathbb{C}$~\cite{helsen2019multiqubit}, which is also appropriate here. We will refer to the above sequence of ``random operation at $A$ - send to $B$- random operation at $B$- send to $A''$ as a \emph{bounce}. The protocol proceeds by performing such a bounce $m$ times, where $m$ is some pre-specified integer. After these $m$ bounces a final operation $G^{(\mathrm{inv})}_A$ is applied at node $A$ after which the state is read out by a two-component POVM $\{E,\id-E\}$. This operation $G^{(\mathrm{inv})}_A$ is not chosen at random but is instead the inverse of the product of all preceding gates, plus some extra ending gate $P_A$, in symbols 
\begin{equation}
G^{(\mathrm{inv})}_A = P_A\left(\prod_{i=1}^m  G^{(i)}_B G^{(i)}_A\right)^{\dagger}.
\end{equation}
This means that if, hypothetically, all gates and state transfer operations are perfectly noise-free the overall operation applied to the initial state $\rho_A$ is the ending gate $P_A$. This ending gate must again be chosen at random, but this time from a restricted gate set of two operations: $P_A\in \{\id, P\}$ where $P$ is a unitary that sends $\rho_A$ to a state orthogonal to $\rho_A$. If $\rho_A$ is the all-zero state a good choice for $P$ would be the all-qubit Pauli $X$-gate. Upon measurement a binary outcome $b$ is produced, which is negated depending on whether $P_A$ is $\id$ or $P$. This is a post-processing trick originally proposed in \cite{helsen2019multiqubit}, making the processing of this output data easier (we will explain this in more detail in \cref{sec:fidelity}). The procedure outlined above must then be repeated for many different random choices of operations, to estimate the average outcome $b_m = \md{E}(b)$. Finally the integer $m$ must be varied, yielding a set of data $\{b_m\}_{m\in \md{M}}$ where $\md{M}$ is some list of integers.

\begin{algorithm}[H]
  \caption{The 2-node network benchmarking protocol}\label{fig:protocol2nodes}
  \begin{algorithmic}[1]
    \For{$m \in \md{M}$}
      \For{$n_m$ from $1$ to $N_m$}
      \State Prepare a state $\rho_A$ at node $A$
        \For{$i$ from $1$ to $m$}
          \State Apply a random gate $G^{(i)}_A$ to $\rho_A$
          \State Transfer $\rho_A$ to node $B$ using $\LAB$
          \State Apply a random gate $G^{(i)}_B$ to $\rho_A$
          \State Transfer $\rho_A$ to node $A$ using $\LBA$
        \EndFor
        \State Choose $P_A$ randomly from the set $\{\id, P\}$
        \State Apply $G^{(\mathrm{inv})}_A = P_A\!\!\left(\!\prod_{i=1}^m  G^{(i)}_B G^{(i)}_A\!\right)^{\dagger}$ to $\rho_A$.
        \State Measure the state $\rho_A$ using the POVM $\{E, \id -E\}$\\ $\hspace{2.7em}$ and record the outcome $b_{n_m}\in \{0,1 \}$
        \If{$P_A$ is equal to $P$}
          \State Set $b_{m_n}$ to  $-b_{m_n}$
        \EndIf
      \EndFor
      \State Compute the mean outcome
      \begin{equation}
        b_m = \frac{1}{N_m}\sum_{n_m=1}^{N_m} {b_{n_m}}
      \end{equation}
    \EndFor
    \State Output the list $\{b_m\}_{m\in \md{M}}$
  \end{algorithmic}
\end{algorithm}

As we will argue in the next section, the output data $\{b_m\}_{m\in \md{M}}$ can be fitted to a single exponential 
\begin{equation}
b_m =_{\mathrm{fit}} Af^m
\end{equation}
where $A$ depends on state preparation and measurement (SPAM) errors and $f$ only depends on the noise incurred by the application of local gates and the channels $\LBA$, $\LAB$. We can extract the quantity $f$ by performing a least-squares fit on the data $\{b_m\}_{m\in \md{M}}$. We will call the quantity $f$ the network link fidelity (associated to node $A,B$). In the next section we will see that, under the assumption that the unitary operations at each node have noise that is the same for each operation, i.e. that there exist quantum channels $\Lambda_A,\Lambda_B$ such that $\Lambda_A^G =\Lambda_A$ and $\Lambda_B^G =\Lambda_B$ for each gate $G$, the network link fidelity $f$ can be written as $f = f(\LAB \Lambda_A)f(\LBA \Lambda_B)$, with $f(\Lambda)$ the depolarizing fidelity (as defined in \cref{eq:fid_to_dep}). This means that $f$ is related to the product of the depolarizing fidelities of $\LAB$ and $\LBA$, but also depends on the local gate noise channels $\Lambda_A,\Lambda_B$.  In practice the local gates will have high fidelity relative to the communication links, so the network fidelity $f$ will be dominated by the channels $\LAB,\LBA$.

\subsection{Multi-node network benchmarking}
The above protocol can be generalized to quantify the fidelity of a connected path of network nodes. This provides a quantum version of the classical `ping' command, and could prove useful in day to day network operation. Consider nodes $A_1, \ldots,A_K$ that are connected by quantum channels $\Lambda_{A_i\rightarrow A_{i+1}}$ and $\Lambda_{A_{i+1}i\rightarrow A_{i}}$ for $i\in \{1,\ldots K\!-\!1\}$. The multi-node network benchmarking protocol works by sending a state from $A_1$ to $A_K$ (along $A_2, A_3,...$) and then back to $A_1$ with a random gate applied to this state at each intermediate node. By performing this multi-node bounce several times one can extract an estimate of the fidelity of the composite link connecting $A_1$ and $A_K$. The protocol is specified in \cref{fig:protocolmultinodes}.
\begin{algorithm}[H]
  \caption{The multi-node network benchmarking protocol}\label{fig:protocolmultinodes}
  \begin{algorithmic}[1]
    \For{$m \in \md{M}$}
      \For{$n_m$ from $1$ to $N_m$}
      \State Prepare a state $\rho_A$ at node $A$
        \For{$i$ from $1$ to $m$}
          \For{$k$ from $1$ to $K-1$}
            \State Apply a random gate $G^{(i,1)}_{A_k}$ to $\rho_A$
            \State Transfer $\rho_A$ to node $A_{k+1}$ using $\Lambda_{A_k\rightarrow A_{k+1}}$
          \EndFor
          \For{$k$ from $K-1$ to $1$}
            \State Apply a random gate $G^{(i,2)}_{A_{k+1}}$ to $\rho_A$
            \State Transfer $\rho_A$ to node $A_{k}$ using $\Lambda_{A_k\rightarrow A_{k+1}}$
          \EndFor
        \EndFor
        \State Choose $P_A$ randomly from the set $\{\id, P\}$
        \State Apply the inverse\\\hspace{1em} $G^{(\mathrm{inv})}_{A_1} = P_A\left(\prod_{i=1}^m \prod_{k=K}^{1} G^{(i,2)}_{A_k}\prod_{k=1}^{K-1} G^{(i,1)}_{A_k}\right)^{\dagger}$ to $\rho_A$.
        \State Measure the state $\rho_A$ using the POVM $\{E, \id -E\}$\\ $\hspace{2.7em}$ and record the outcome $b_{n_m}\in \{0,1 \}$
        \If{$P_A$ is equal to $P$}
          \State Set $b_{m_n}$ to  $-b_{m_n}$
        \EndIf
      \EndFor
      \State Compute the mean outcome
      \begin{equation}
        b_m = \frac{1}{N_m}\sum_{n_m=1}^{N_m} {b_{n_m}}
      \end{equation}
    \EndFor
    \State Output the list $\{b_m\}_{m\in \md{M}}$
  \end{algorithmic}
\end{algorithm}
As in the $2$-node case, the output data $\{b_m\}_{m\in \md{M}}$ can be fitted to a single exponential 
\begin{equation}
b_m =_{\mathrm{fit}} A_{\mathrm{SPAM}}f^m
\end{equation}
Again assuming gate-independent noise for the local gates at each node, we can see (in an argument identical to that in the $2$-node case) that $f$ will be given by
\begin{equation}
f = \prod_{k=1}^{K-1} f(\Lambda_{A_k\rightarrow A_{k+1}} \Lambda_{A_k}) f(\Lambda_{A_{k+1}\rightarrow A_{k}} \Lambda_{A_{k+1}}).
\end{equation}
In other words the output of the multi-node network benchmarking protocol is given by the product of the depolarizing fidelities of all intermediate communication links (up to local noise channels). We will refer to $f$ as the network path fidelity (associated to the path $A_1, \ldots ,A_K$).

\section{Network fidelity}\label{sec:fidelity}
In this section we will argue that the 2-node network benchmarking protocol proposed in \cref{fig:protocol2nodes} yield an output related to the product of the fidelities of the maps $\LAB$ and $\LBA$. This argument will easily generalize to the multi-node case. We will for this section assume that the network obeys the property of gate-independent noise. This means we assume that a gate $G_A$ acts as $\Lambda^G_A(G_A\rho_AG_A\ct)$ and $\Lambda^G_B(G_B\rho_BG_B\ct)$ for all $G \in \md{G}$. The arguments given here are closely related to those for randomized benchmarking \cite{PhysRevA.75.022314,knill2008randomized} and subsequently the assumptions we make can be relaxed significantly by adapting the more modern treatments of randomized benchmarking~\cite{IndependentNoise,Merkel18,helsen2020general} to the network benchmarking setting, but we will not pursue this here. \\

Consider the average outcome $b_m$ of an $m$-bounce sub-protocol (for some $m \in \md{M}$), as given in \cref{fig:protocol2nodes}. This average outcome can be written out as
\begin{align}\label{eq:av_prob}
  b_m &= \mathlarger{\mathlarger{\md{E}}} \bigg(\!\tr\Big[\Lambda_A^{M}(E)\Big[\Lambda_A[\id - P]\Big(\prod_{i=1}^m  G^{(i)}_B G^{(i)}_A\Big)^{\dagger} \notag\\&\hspace{5em} \times\LBA \Lambda_B G^{(m)}_B\LAB \cdots \Lambda_B \notag\\&\hspace{9em}\times G^{(1)}_B \LBA  \Lambda_A G^{(1)}_A \big](\rho)\Big]\!\bigg),
\end{align}
where the average is taken independently over $G_A^{(1)},\ldots, G_B^{(m)}$.
We can rewrite this quantity into something more manageable. Note first that, by linearity and independence, we can move the average over $G_B^{(m)}$ into the trace. Here we recognize the twirl operator
\begin{equation}
T(\LAB\Lambda_A) = \frac{1}{|\md{G}|}\sum_{G_B^{(m)}\in \md{G}} {G_B^{(m)}}\ct \LAB\Lambda_A G_B^{(m)}
\end{equation}
Now we can use the fact that $\md{G}$ is a $2$-design to conclude that this twirl operator $T(\LAB\Lambda_A)$ is a depolarizing channel with depolarizing fidelity $f (\LAB \Lambda_A)$~\cite{nielsen2002simple}. Further using the fact that a depolarizing channel commutes with unitary operations we can perform this same trick for for the remaining random gates $G_A^{(1)},\ldots,G_B^{(m-1)},G_B^{(m)}$ to obtain 
\begin{align}
  b_m &= \tr \bigg(\Lambda_A^{M}(E)\left[T(\LAB\Lambda_A)T(\LBA\Lambda_B)\right]^m\notag\\&\hspace{7.5em}(\Lambda_A^{SP}(\rho_A) -P\Lambda_A^{SP}(\rho_A)P\ct )\bigg).
\end{align}
Next we note that $\tr(\Lambda_A^{SP}(\rho_A) -P\Lambda_A^{SP}(\rho_A)P\ct ) = 0$ (by cyclicity of the trace). Together with the fact that $[T(\LAB\Lambda_A)T(\LBA\Lambda_B)]^m$ is a depolarizing channel with depolarizing fidelity $[f (\LAB \Lambda_A)f (\LBA \Lambda_B)]^m$ this allows us to conclude that
\begin{equation*}
  b_m \!= \!\tr\!\big[(E\!-\!P\ct \!EP)\mc{E}_A(\rho)\big]\!\big[f(\LAB\Lambda_A)f(\LBA \Lambda_B)\big]^m.
\end{equation*}
Hence we can obtain an estimate of the product $f (\LAB \Lambda_A)f (\LBA \Lambda_B)$ by fitting the data $\{b_m\}_{m\in \md{M}}$ to the formula
\begin{equation}
 b_m =_{\mathrm{fit}} A_{\mathrm{SPAM}}f^{m}.
\end{equation}
We can make a similar argument for the multi-node protocol, where we conclude that the average data $\{b_m\}_{m\in \md{M}}$ can be described as
\begin{equation*}
  b_m =_{\mathrm{fit}} A_{\mathrm{SPAM}}\big[f_{A_1A_2}\ldots f_{A_{K-1}A_K} f_{A_KA_{k-1}} \ldots f_{A_2A_1}\big]^{m},
\end{equation*}
with $f_{A_{i-1}A_i} = f(\Lambda_{A_{i-1} \rightarrow A_i}\Lambda_{A_i-1})$ where $\Lambda_{A_i-1}$ is the quantum channel modeling (gate-independent) local noise in the node $A_i$, and similarly for $f_{A_{i}A_{i-1}}$.\\

\subsection{Symmetric fidelity and teleportation}
The 2-node network benchmarking protocol gives an estimate of the product of the depolarizing fidelities of the channels $\LAB$ and $\LBA$ modeling the links between node $A$ and node $B$ (up to local operation noise). However, in some relevant cases the channels $\LBA$ and $\LAB$ have equal average fidelity (and thus depolarizing fidelity), in which case this average fidelity is directly accessible through network benchmarking. Here we discuss one important case where this is true, namely when the channels $\LAB$ and $\LBA$ are implemented through the quantum teleportation protocol using some pre-prepared entangled state $\rho_{AB}$ (note that this is not necessarily a perfect maximally entangled state) between nodes $A$ and $B$. Concretely we will prove the following lemma:

\begin{lemma}
Let $\LAB$ be the quantum channel implemented by teleportation using a state $\rho_{AB}$ as a resource, and let $\LBA$ be the quantum channel implemented by teleportation using a state $\rho_{BA}$ as a resource. If the local operations used in the teleportation process are noiseless, then $F(\LAB) = F(\LBA)$.
\end{lemma}
\begin{proof}
We begin by noting that the average fidelity of any quantum channel $\Lambda$ is related to its entanglement fidelity $\bra{\Phi}\id\otimes \Lambda(\Phi)\ket{\Phi}$ where $\Phi$ is the maximally entangled state, as (from \cite{nielsen2002simple})
\begin{equation}\label{eq:av_to_ent_fid}
F(\Lambda) = \frac{d (\bra{\Phi}\id\otimes \Lambda(\Phi)\ket{\Phi}) +1}{d+1} = \frac{d F_e(\Lambda) +1}{d+1}.
\end{equation}
Next we use a result from \cite[equation 25]{horodecki1999general} stating that the entanglement fidelity of a channel $\LAB$ induced by teleportation (with perfect local operations) with a state $\rho_{AB}$ is equal to the singlet fraction $F_s(\rho_{AB}) = \bra{\Phi} \rho_{AB}\ket{\Phi}$ of the state $\rho_{AB}$. Similarly we have that $F_e(\LBA) = F_s(\rho_{BA})$. Now noting that the singlet fraction is invariant under the interchange of A and B we have $F_e(\LAB) = F_s(\rho_{AB}) = (F_s(\rho_{BA}) =F_e(\LBA) $ and thus $F(\LAB) = F(\LBA)$, which proves the lemma.

\end{proof}
In this case we can thus connect the network fidelity $f$, as measured by the $2$-node protocol, to the average fidelity of the network links $A\to B$ and $B\to A$ (assuming negligible contributions from local noise). We have
\begin{align}
\sqrt{f}  &= \sqrt{f (\LAB )f (\LBA)} = f(\LAB) = f(\LBA) \notag\\&= \frac{dF_{\mathrm{avg}}(\LAB) -1}{d-1},
\end{align}
where we used \cref{eq:fid_to_dep}.

\section{Simulation results}\label{sec-sim}
\begin{figure*}
\centering
\includegraphics[scale = 0.55]{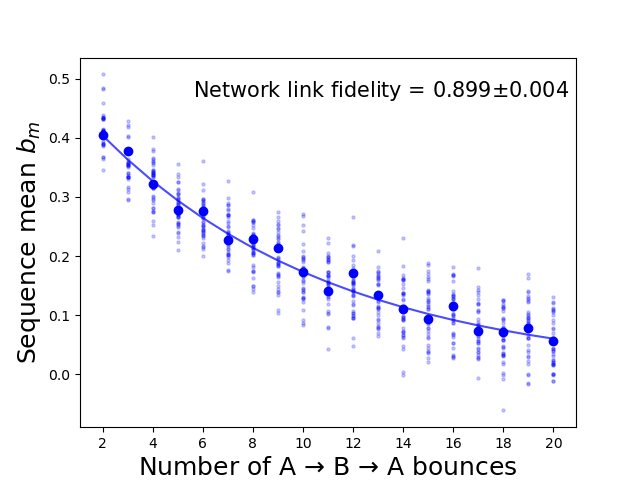}\hspace{0em}\includegraphics[scale=0.55]{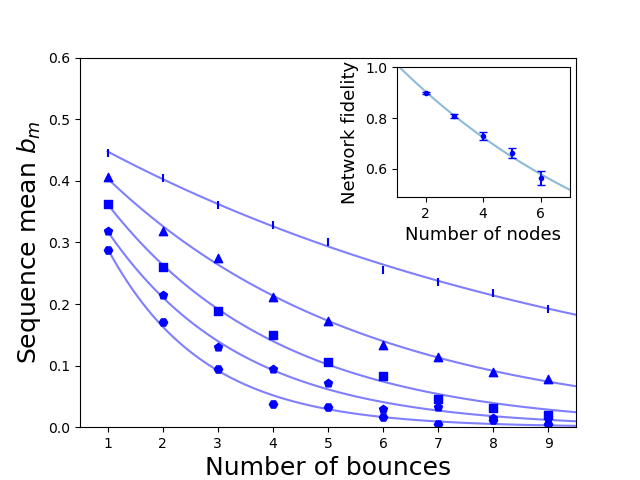}
\caption{\textbf{(a)}: Simulation in Netsquid of the two node network benchmarking protocol (\cref{fig:protocol2nodes}). The nodes $A$ and $B$ hold qubits afflicted by dephasing $(T2)$ errors as seen in NV-center quantum processors~\cite{bradley2019ten}. The channels $\LAB$ and $\LBA$ emulate teleportation with a noisy quantum state of the form \cref{eq:tel_noise}, with a bright state population of $\alpha=0.95$. \textbf{(b)}: Simulation in Netsquid of the multi-node network benchmarking protocol (going from two (line markers) to six (hexagon markers) nodes in a linear configuration). The noise models for network links and node operations are as before. We observe that the network fidelity decays exponentially with the number of nodes. }\label{fig:sim_results}
\end{figure*}
In this section we discuss the results of a simulation of network benchmarking on a model network using the quantum network simulator NetSquid~\cite{netsquid}. Netsquid is an advanced discrete event simulator that allows for the testing of quantum network properties in realistic circumstances, taking into account noisy operations and state preparation and measurement errors, but also issues specific to networks such as delay-induced decoherence, packet loss, and protocol timing issues. The code that generates the results below can be found at~\cite{zenodolink}. The goal of this section is to show how network benchmarking can be applied in practice. To this end we have constructed two different simulations inspired by real world scenarios. The first simulation investigates the behavior of $2$-node network benchmarking in a scenario where two network nodes are connected by network links implemented by teleportation, and the second investigates the use of the multi-node protocol as an efficient method to detect the decay of quality as the number of nodes in a path increases. For both these simulations we choose the local gate set to be the single qubit Clifford group $\mathbb{C}_1$. For both simulations we also specify a noise model that is an abstracted and simplified version of the noise present in networks based on NV-centres~\cite{humphreys2018deterministic}. We however emphasize that our intent is not to produce a detailed physical simulation of networks of this form (we do not take into account e.g. waiting times and non-deterministic entanglement creation), but rather to gain intuition for the behaviour of the network benchmarking protocol. Specifically, in both simulations we will model physical $T1$ and $T2$ noise affecting qubits in the network nodes, model the quantum network links with a quantum channel emulating state transfer through teleportation (explained in more detail below), and omit other imperfections.
\subsection{Teleportation mediated link between two nodes}
\Cref{fig:sim_results} (a) show the outputs of 2-node network benchmarking, as simulated in Netsquid. The links connecting $A$ and $B$ are here modeled by teleportation using a noisy entangled state of the form
\begin{equation}\label{eq:tel_noise}
\rho_{AB} = \alpha |\Phi\rangle\!\langle \Phi| + (1-\alpha)|00\rangle\!\langle 00|,
\end{equation}
where $\ket{\Phi}$ is again the maximally entangled state and $\alpha$ is the bright state population of the qubit at the NV network node before entanglement generation.
This state arises as a well-motivated model of single-photon heralded entanglement generation in NV centers~\cite{humphreys2018deterministic}. For our simulation we choose a bright state population of $\alpha=0.95$ (slightly different from the value in \cite{humphreys2018deterministic}). Moreover we model the qubits in the local nodes as being afflicted by standard $T_2$ dephasing noise, with relevant values for ${ }^{13}C$ memory qubits in NV-center quantum processors being $T_2 =12$ms (we technically also include $T1$ amplitude damping noise, however this is not a critical factor in NV centers~\cite{bradley2019ten}). Correspondingly we assume that applying native quantum operations on these memory qubits takes $39\mu$s (see \cite[figure $5$]{bradley2019ten} for the above numbers). We note that since some of the gates in the single qubit Clifford group must be compiled out of native operations this is not a gate-independent noise model. The data in \cref{fig:sim_results} is generated by running the $2$-node network benchmarking protocol for $40$ random sequences for each number of bounces $m$ (ranging from $1$ to $20$). Netsquid tracks density matrices, and can thus calculate the mean outcome for a random sequence directly. We add Gaussian noise to the data to simulate shot noise for $4000$ measurements per random sequence. The mean outcome for each random sequence of local gates is shown in light-blue, and the average over all sequences is shown in dark blue. From the exponential decay fit we obtain $f = 0.899 \pm 0.004$ ($95\%$ Studentized confidence interval from the fit) which is in line with a fidelity dominated by the quality of the teleportation procedure.

\subsection{Teleportation mediated links between multiple nodes nodes}
\Cref{fig:sim_results} (b) show the outputs of multi-node network benchmarking, as simulated in Netsquid. In this simulation we performed multi-node network benchmark on $n$ nodes in a linear configuration, where $n$ ranges between $2$ and $6$, with the links connecting the node modeled again by teleportation using the same parameters as before. The data in \cref{fig:sim_results} (b) is generated by running the $n$-node network benchmarking protocol for $40$ random sequences for each number of bounces $m$ (ranging from $1$ to $9$). From this we can infer that the network fidelity decreases from $0.899\pm 0.04$ at two nodes (line markers), to $0.56\pm 0.02$ at six nodes (hexagon markers, $95\%$ Studentized confidence interval from the fit). We note that empirically the network fidelity decays exponentially with the number of nodes (see inset in \cref{fig:sim_results} (b)). This points to a potential use of network benchmarking as a network discovery tool, in this case to give heuristic estimates of upper limits on the distance quantum information can travel through a network before degrading, without having to necessarily explore the whole network. 

\section{Statistics of network benchmarking}\label{sec:stat-anal}
In this section we analyze the finite sampling properties of the network benchmarking protocol. This analysis will resemble earlier statistical analyses of standard randomized benchmarking \cite{helsen2019multiqubit,harper2019statistical}, with one key difference. In standard analyses the accuracy of the fidelity estimate is given as a function of the number of measurements that must be performed. This ignores that some measurements might be more expensive to perform than others. In particular, one typically assumes that it is not more costly to obtain a sample from a long sequence of gates than it is to obtain a sample from a short sequence of gates.\\

In network benchmarking however, this assumption is no longer reasonable, as the cost of transmitting a qubit over a long distance will be the dominant factor in the cost of a sample. Hence sampling a sequence containing $m$ bounces (as specified in \cref{fig:protocol2nodes} and \cref{fig:protocolmultinodes}) will be approximately $m$ times as expensive as sampling a sequence containing only one bounce. This means it is more appropriate to estimate accuracy of the fidelity estimate produced by \cref{fig:protocol2nodes} as a function of the number of bounces. Taking this cost into account has strong consequences for the statistical properties of network benchmarking. In particular we will argue that we can not achieve `multiplicative accuracy' for the estimation of fidelity when taking the number of state transmissions as a cost metric. However, as seen in the simulations in \cref{sec-sim}, network benchmarking achieves good statistical accuracy for a reasonable resource use in practice. Moreover in the immediate future network fidelities are expected to be reasonably low (in the $90 {\raise.17ex\hbox{$\scriptstyle\mathtt{\sim}$}}
 99\%$ regime), so additive, and not multiplicative, accuracy is enough for practical purposes. 

\subsection{Relative accuracy estimation}
One of the main selling points of standard randomized benchmarking is its ability to estimate the infidelity $r= 1-f$ where $f$ is the depolarizing fidelity measured by randomized benchmarking, to \emph{multiplicative precision}. This means the estimator $\hat{r}$ is distributed around its true value $r$ with variance $O(r^2)$ \cite{helsen2019multiqubit,harper2019statistical}, which means that estimation in the high fidelity regime ($r<<<1$) is not more costly than estimation in the low fidelity regime. We will argue here that this behavior is critically dependent on the assumption that the cost of obtaining samples from a given gate sequence in a (network) benchmarking experiment is independent of the sequence length. As discussed above this is a reasonable assumption for standard randomized benchmarking but not so much for network benchmarking. We point out however that the argument below works just as well for standard randomized benchmarking if one takes the number of gates implemented as a cost function (as opposed to the number of samples collected). The argument below is not strictly rigorous as we will be making standard statistical assumptions such as normality of distributions, but we expect it can be made rigorous with sufficient work. \\

In $2$-node network benchmarking we can define the network infidelity as $r  = 1-f$. Network benchmarking constructs an estimator $\hat{r}$ for $r$ by sampling the decay function $Af^m$ for different sequence lengths $m$ and then fitting an exponential through the resulting averages. Without loss of generality we can assume the parameter $A$ to be known, as perfect knowledge of a parameter in an estimation problem will never increase the difficulty of estimating another parameter. Now our goal is to give a lower bound on the estimation cost of $f$, given samples from distributions $D(f, m)$ with mean $A f^m$ and variance $V(f,m)$. This distribution $D(f,m)$ is the distribution sampled by executing steps $2-17$ in \cref{fig:protocol2nodes}. We will make an argument using the Cramer-Rao bound, which states that the variance of any unbiased estimator of $f$ must be larger than the inverse of the Fisher information, defined as
\begin{equation}
I(f) = \frac{Af^{2m-2}m^2}{\md{V}(f,m)},
\end{equation}
for some fixed $m$, where we assumed that $D(f,m)$ is a Gaussian. This is a reasonable assumption since $D(f,m)$ is defined as the distribution of the mean of many independent random variables. 
 The central parameter that determines the Fisher information and thus the estimation cost is the variance of $\md{V}(f,m)$ of $D(f,m)$. By the law of total variance we decompose the variance of $D(f,m)$ into three contributions
\begin{equation}
\md{V}(f,m) = \md{V}_g(f,m) + \md{V}_{meas}(f,m) + \md{V}_{diff}(f,m),
\end{equation}
where $\md{V}_g(f,m)$ is the variance due to the randomness in selecting a sequence of gates, $\md{V}_{meas}$ is the variance due to the estimation of the probability $p(\vec{G})$ (this is often called shot-noise in the experimental literature) and $\md{V}_{diff}$ is the variance associated to the random choice of Pauli operator at the end of each sequence. $V_g(f,m)$ and $V_{meas}(f,m)$ have dependencies on both $f$ and $m$, making analysis difficult. However we can lower bound both by zero (which never makes the inference task harder) and state that $V(f,m)\leq \md{V}_{diff}(f,m)=1/(4\cdot2)$, where the factor of $1/2$ is due to the division by one half in step $17$ of the 2-node protocol \cref{fig:protocol2nodes}. We can thus upper bound the Fisher information $I(f,m)$ of $f$ in the distribution $D(f,m)$ as
\begin{equation}
I(f,m) \leq 8A^2 m^2 f^{2m-2}.
\end{equation}
This is the Fisher information associated to a fixed sample. We can consider the Fisher information associated to the sampling cost (which grows linearly with the sequence length $m$) by dividing by $m$, to get
\begin{equation}
I_{cost}(m,f) \leq 8A^2 m f^{2m-2}.
\end{equation}
Now given that we want to lower bound the variance of the estimator we are interested in the maximum of $I_{cost}(m,f)$ over $m$. It can be easily seen that this function has a unique maximum at $m = \frac{-1}{2\log(f)}$. This means the maximal Fisher information is
\begin{equation}
I_{cost,max}(f) \leq \frac{-4A^2}{log(f)} f^{1/\log{f} -2}.
\end{equation}
Writing $f = 1-r$ and writing out the Mercator series for the logarithm we can see now that $I_{cost,max}(f) = O(r)$. This implies through the Cramer-Rao bound that 
\begin{equation}
\md{V}(\hat{r}) = O(r),
\end{equation}
providing additive, but not relative, estimation accuracy. 

\section{Conclusion}
In this paper we have presented the network benchmarking protocol, a robust and efficient tool for assessing the quality of network links between nodes in a quantum Internet. We gave two version of the protocol, a $2$-node version, analyzing the quality of a single connection, and a multi-node version, analyzing the quality of a path of nodes in a network. We gave a mathematical analysis of these protocols, arguing that under some assumptions they output a quantity related tot the average fidelity of the quantum channels modeling the network links. We also argued that for a standard class of network link models, namely noisy quantum teleportation, the network fidelity can be exactly related to the average fidelity of the link. We supplemented this theoretical work with numerical simulations using the quantum network simulator Netsquid. From these simulations we saw that network benchmarking works well in realistic environments. A natural next step would be to implement the network benchmarking protocol in real quantum networks, which are currently in development. On the theoretical side it would be interesting to investigate further the use of network benchmarking as a tool for network discovery, integrating it as a subroutine in online routing algorithms for quantum networks~\cite{chakraborty2019distributed}, which will have to take the quality of a network link into account when making routing decisions. 

\begin{acknowledgments}
JH would like to thank Bas Dirkse, Harold Nieuwboer, and Michael Walter for fruitful conversations. This work was funded by an ERC Starting Grant (S.W.) and the NWO Quantum Software Consortium.
\end{acknowledgments}

\bibliographystyle{unsrt}

\bibliography{telbench}
\end{document}